\documentclass [reqno] {amsart}
\usepackage{tikz}
\usepackage{hyperref}

\copyrightinfo{2000}{American Mathematical Society}

\newtheorem{theorem}{Theorem}[section]
\newtheorem{lemma}[theorem]{Lemma}
\newtheorem{corollary}[theorem]{Corollary}

\theoremstyle{definition}

\theoremstyle{remark}
\newtheorem{remark}[theorem]{Remark}

\numberwithin{equation}{section}

\newcommand{\C}{\mathbb{C}}
\newcommand{\iD}{\mathit{\Delta}}
\newcommand{\R}{\mathbb{R}}
\newcommand{\cS}{\mathcal{S}}
\newcommand{\cX}{\mathcal{X}}

\begin{document}

\title[Hyperbolic Formulas in Elliptic Cauchy Problems]
      {Hyperbolic Formulas in Elliptic Cauchy Problems}

\author{D. Fedchenko}

\address[Dmitry Fedchenko]
        {Institute of Mathematics,
         Siberian Federal University,
         Svobodny Prospekt 79,
         660041 Krasnoyarsk,
         Russia}

\email{dfedchenk@gmail.com}


\author{N. Tarkhanov}

\address[Nikolai Tarkhanov]
        {Institute of Mathematics,
         University of Potsdam,
         Am Neuen Pa\-lais 10,
         14469 Potsdam,
         Germany}

\email{tarkhanov@math.uni-potsdam.de}

\date{February 7, 2010}




\begin{abstract}
We study the Cauchy problem for the Laplace equation in a cylindrical domain
with data on a part of it's boundary which is a cross-section of the cylinder.
On reducing the problem to the Cauchy problem for the wave equation in a
complex domain and using hyperbolic theory we obtain explicit formulas for the
solution, thus developing the classical approach of Hans Lewy (1927).
\end{abstract}

\maketitle

\tableofcontents

\section*{Introduction}
\label{s.Introduction}

The question of the well-posedness of the Cauchy problem was first raised by
Hadamard who proved in \cite{Hada23} that it is ill-posed in the case of
linear second order elliptic equations.
Hadamard's proof is based on the analytic regularity of linear boundary value
problems.
This regularity has been extended to nonlinear elliptic equations in
   \cite{Morr58}
so that Hadamard's argument also applies to general nonlinear elliptic
equations.

Hadamard also pointed out in \cite{Hada23} that the problem occurring in wave
propagation is not at all analytic problem, but a problem with real, not
necessarily analytic data.
For general linear equations it is well known that the hyperbolicity is a
necessary condition for the well-posedness of the noncharacteristic Cauchy
problem in $C^\infty$, that is for the existence of solutions for general
$C^\infty$ data, cf. \cite{Lax57}, \cite{Mizo61}.
Moreover, for several classes of nonhyperbolic equations, explicit conditions
on the initial data necessary for the existence of solutions were given in
   \cite{Nish84}.
For nonlinear equations, \cite{Waka01} proves that the existence of a smooth
stable solution implies hyperbolicity, stability meaning that one can perturb
the initial data and the source terms in the equations.

The nonlinear theory yields difficult new problems, see \cite{HounFilh92},
                                                        \cite{Meti06},
                                                        etc.
There are many interesting examples, for instance in multiphase fluid dynamics,
where the equations are nor everywhere hyperbolic.
As but one occurrence of this phenomenon, we consider Euler's equations of gas
dynamics in Lagrangian coordinates
\begin{equation}
\label{eq.eegd}
   \left\{
   \begin{array}{rcl}
     \partial_t u + \partial_x v
   & =
   & 0,
\\
     \partial_x p (u) + \partial_t v
   & =
   & 0
   \end{array}
   \right.
\end{equation}
mentioned in \cite{Meti06}.
The system is hyperbolic, when $p' (u) > 0$, and elliptic, when $p' (u) < 0$.
For van der Waals state laws, it happens that $p$ is decreasing on an interval
   $[u_\ast, u^\ast]$.
A mathematical example is $p (u) = u (u^2 - 1)$.
Hadamard argument shows that the Cauchy problem with data taking values in
the elliptic region is ill-posed.
If
   $u (0,x) = u_0 (x)$ is real analytic near $x$ and
   $u_0 (x)$ belongs to the elliptic interval,
then any local $C^1$ solution is analytic, see e.g. \cite{Morr58}.
Thus, the initial data $u_0 (x)$ must be actually analytic for the initial
value problem to have a solution.

It was Hans Lewy who first used hyperbolic techniques to study problems for
elliptic equations, cf. \cite{Lewy29}.
The solutions of elliptic equations with real analytic coefficients prove to
be real analytic, and so they extend to holomorphic functions in a complex
neighbourhood of their domain.
For a holomorphic function obtained in this way the derivative
   $\partial / \partial x_k$
just amounts to the derivative
   $\partial / \partial (\imath y_k)$
where $z_k = x_k + \imath y_k$ are complex variables with $k = 1, \ldots, n$.
One can go to a complex space in only one variable, say $x_n$, and the change
   $\partial / \partial x_n \mapsto - i \partial / \partial y_n$
leads to a drastical modification of the characteristic variety.
The Laplace equation written in the coordinates $(x',x_n)$ with
   $x' = (x_1, \ldots, x_{n-1})$
transforms to the wave equation in the coordinates $(x',y_n)$.

This idea is especially useful in the study of the Cauchy problem for elliptic
equations.
This problem is overdetermined even in the case of data given on an open part
of the boundary, hence it does not admit any simple formulas for solutions,
   see however \cite{Yar75} and \cite{Shl92}.
Since the problem is unstable, the left inverse operator fails to be
continuous.
On the other hand, the Cauchy problem for hyperbolic equations is of textbook
character and it admits many explicit formulas for solutions like
   d'Alembert,
   Kirchhoff,
   Poisson,
etc. formulas, cf. \cite{Hada23}.
Outstanding contribution to the Cauchy problem for hyperbolic equations is due
to Leray who developed multidimensional residue theory in complex analysis to
handle the problem, see \cite{Lera57},
                        \cite{Lera63},
etc.
Having granted a solution $u (x',\imath y_n)$ of the Cauchy problem for a
hyperbolic equation, how can one restore the solution $u (x',x_n)$ of the
Cauchy problem for the original elliptic equation?
The simple substitution $\imath y_n \mapsto x_n$ does not make sense in
general.
For this purpose we invoke a formula of \cite{Carl26} which restores the
values of holomorphic functions in a corner on the diagonal through their
values on an arc connecting to faces of the corner.
The resulting formula for the solution of an elliptic Cauchy problem includes
a limit passage and agrees perfectly with the general observation that the
character of instability in an elliptic Cauchy problem is similar to that in
the problem of analytic continuation, cf. \cite{Tark95}.

As mentioned, the idea to use hyperbolic formulas for elliptic Cauchy problems
goes back at least as far as \cite{Lewy29}.
In the 1960s it was directly applied in a number of papers by Krylov, see for
instance \cite{Kryl69}.
In \cite{Kryl69}, an integral representation for holomorphic solutions of a
partial differential equation in a complex domain is constructed through the
Cauchy data of solutions on an analytic surface.
However, the formula does not manifest any instability of the Cauchy problem,
which shows its local character.

The approach we develop in this paper has the advantage of providing a large
parameter to perturb the solution of the problem.
This might give rise to a calculus of Cauchy problems for elliptic equations.
Since these problems are unstable, no operator calculus similar to that
including elliptic boundary values problems and their parametrices on compact
manifolds with boundary is possible.
On introducing a large parameter into operators we are able to describe their
perturbations which lead to solutions.

Let us dwell on the contents of the paper.
In Section \ref{s.tcp} we formulate the Cauchy problem for a second order
elliptic equation in a domain $\cX$ in $\R^n$.
The principal part of the equation is given by the Laplace operator while the
lower order part may include nonlinear terms.
The Cauchy data are given on a nonempty open set $\cS$ of the boundary.
Our standing assumption is that
   $\cX$ is a cylinder over a bounded domain $B$ with smooth boundary
   in the space $\R^{n-1}$ of variables $x'$
and
   $\cS$ a smooth cross-section of $\cX$.

In Section \ref{s.hr} we reformulate the same Cauchy problem for a hyperbolic
equation.
Namely, we assume that the solution $u (x',x_n)$ is a real analytic function
of $x_n \in (b (x'),t (x'))$ for each fixed $x' \in B$.
Then it extends to a function $u (x',z_n)$ holomorphic in a narrow strip
   $- \varepsilon < y_n < \varepsilon$
around the interval $(b (x'),t (x'))$ in the plane of complex variable
   $z_n = x_n + \imath y_n$.
The Cauchy-Riemann equations force $u (x',z_n)$ to fulfill
   $(\partial / \partial x_n) u = - \imath (\partial / \partial y_n) u$
in the strip $(b (x'),t (x')) \times (- \varepsilon,\varepsilon)$.
Hence, we rewrite the original elliptic equation as a hyperbolic equation for
a new unknown function of variables $(x',y_n)$.
Since $\cS$ is the graph of some smooth function $x_n = t (x')$ on $B$,
the Cauchy data transform easily for the new unknown function.

In Section \ref{s.tpc} we test our approach in the case of two variables.
It is precisely the case treated in \cite{Lewy29}, and the approach of
                                    \cite{Lewy29}
does not work for $n > 2$.
For $n = 2$, the geometric picture is especially descriptive because the
complexification of $x_2$ does not lead beyond $\R^3$.

On solving the Cauchy problem for a hyperbolic equation in a conical domain in
the space of variables $(x',y_n)$, we are left with the task of continuing the
solution given on the base of an isosceles triangle analytically along the
bisectrix of the angle at the vertex, for each fixed $x' \in B$.
To this end we invoke the classical formula of Carleman established precisely
for this configuration, see \cite{Carl26}.
Of course, the use of Carleman's formula is justified only for real analytic
solutions of the original elliptic Cauchy problem.
In Section \ref{s.cf} we give a simple proof of this formula.
Numerical simulations with Carleman's formula failed to manifest its striking
efficiency.
However, nowadays more efficient formulas of analytic continuation are
available, cf. \cite{Aize93}.

In Section \ref{s.pf} we investigate the Cauchy problem for the inhomogeneous
Laplace equation in the space $\R^n$ of variables $(x',x_n)$ with odd $n$.
As is shown in Section \ref{s.hr}, it reduces to the Cauchy problem for the
inhomogeneous wave equation in the space of variables $(x',y_n)$.
The case $n = 1$ deserves a special study, for it concerns the initial problem
for ordinary differential equations.
If $n = 3$, the Cauchy problem for the wave equation possesses a very explicit
solution constructed by Poisson.
For odd $n \geq 5$ an explicit solution formula was derived by Hadamard in
   \cite{Hada23}
by his method of descent.
On substituting it into Carleman's formula and changing integrations over $y_n$
                                                                     and $x'$,
we get a formula for solutions of the Cauchy problem for harmonic functions.

In Section \ref{s.kf} we restrict our attention to the Cauchy problem for the
inhomogeneous Laplace equation in the space $\R^n$ of variables $(x',x_n)$
with even $n$.
By the above it reduces to the Cauchy problem for the inhomogeneous wave
equation in the space of variables $(x',y_n)$.
The latter Cauchy problem admits a very explicit solution formula due to
   d'Alembert in the case $n = 2$ and
   Kirchhoff in the case $n = 4$.
For general even $n$ the formula seems to be first published in \cite{Hada23}.
We combine it with Carleman's formula and change the integration over $y_n$
and over $x'$.
This yields an explicit formula for solutions of the Cauchy problem for the
inhomogeneous Laplace equation.
To our best knowledge, this formula has never been published.

In Section \ref{s.cr} we analyse if our approach applies to Cauchy problems
for elliptic equations of order different from two.
Yet another question under study is whether the method of quenching functions
in the Cauchy problem for the Laplace equation presented in \cite{Yar75} is
actually a very particular case of formulas elaborated in this paper.

\section{The Cauchy problem}
\label{s.tcp}

Let $\cX$ be a bounded domain with piecewise smooth boundary in $\R^n$.
We require $\cX$ to be of cylindrical form, i.e., $\cX$ is a part of the
cylinder $B \times \R$ intercepted by two surfaces $y_n = b (x')$
                                               and $y_n = t (x')$
over $B$, where
   $B$ is a bounded domain with smooth boundary in the space $\R^{n-1}$ of
   variables $x' = (x_1, \ldots, x_{n-1})$.
For simplicity we assume that $t (x') > b (x')$ for all $x' \in B$, the case
   $t (x') = b (x')$ for some or all $x' \in \partial B$ is not excluded.
The Cauchy data will be posed on the top surface
   $\cS := \{ (x',t (x')) : x ' \in B \}$
which is tacitly assumed to be real analytic,
   cf. Fig.~\ref{f.cylinder}.
\begin{center}
\begin{figure}[h]
\begin{picture}(120,120)

\begin{tikzpicture}
\draw[->] (0,0) --(3.5,0) node[right] {$ $};
\draw[->] (0,0) -- (0,2.7) node[above] {$x_n$};
\draw[->] (0,0) -- (-0.7,-0.7) node[right] {$ $};

\draw[dashed] (0,2.57) -- (1.25,2.17);

\draw (2,1.5) -- (2,0.5);
\draw (0.5,1.5) -- (0.5,0.5);
\draw[dashed] (2,0.5) -- (2,-0.5);
\draw[dashed] (0.5,0.5) -- (0.5,-0.5);

\draw (2,-0.5) arc (0:360:0.75cm and 0.25cm) node[right=1pt, fill=white] {$B$}(0pt,0pt);

\draw[dashed] (2,0.5) arc (0:180:0.75cm and 0.25cm); 
\draw (0.5,0.5) arc (180:360:0.75cm and 0.25cm) node[right=1pt, fill=white] {$x_n=b(x')$}(0pt,0pt);

\draw[color=red, dashed] (2,1.5) arc (0:180:0.75cm and 0.25cm);
\draw[color=red] (0.5,1.5) arc (180:360:0.75cm and 0.25cm)node[right=1pt, fill=white] {$x_n=t(x')$}(0pt,0pt);

\draw[color=red, line width=0.7pt] (2,1.5) .. controls
(1,2.5)  .. (0.5,1.5) node[above=1pt, fill=white] {$\mathcal{S}$}(0pt,0pt);

\draw[dashed] (2,0.5) .. controls
(1.35,-0.25)  .. (0.5,0.5);

\filldraw [black] (1.25,-0.5) circle (1pt) node[below=10pt, fill=white] {$x'$}(0pt,0pt);
\filldraw [black] (1.25,-0.05) circle (1pt);
\filldraw [red] (1.25,2.17) circle (1pt);
\draw[dashed] (0,0) -- (1.25,-0.5);
\draw[dashed] (0,0.45) -- (1.25,-0.05);

\draw[dashed] (1.25,-0.5) -- (1.25,2.17);

\end{tikzpicture}

\end{picture}
\caption{A typical domain under consideration.}
\label{f.cylinder}
\end{figure}
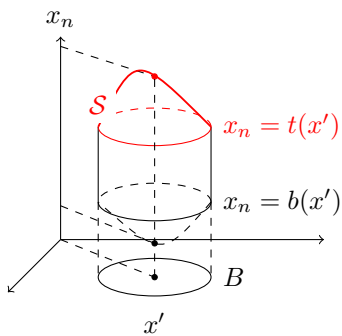
\end{center}

For an elliptic second order differential operator on the closure of $\cX$ the
Cauchy data on $\cS$ look like
$$
\left\{ \begin{array}{rclcl}
          u
        & =
        & u_0
        & \mbox{on}
        & \cS,
\\
          \displaystyle
          \frac{\partial u}{\partial \nu}
        & =
        & u_1
        & \mbox{on}
        & \cS,
        \end{array}
\right.
$$
where
   $\nu$ is the outward unit normal vector at $\cS$.
Obviously,
   $\nu = \nabla \varrho / |\nabla \varrho|$
where $\varrho = x_n - t (x')$.

\begin{lemma}
\label{l.Neumann}
If $u$ is a smooth function near $\cS$ satisfying $u = u_0$ on $\cS$, then
$$
   \frac{\partial u}{\partial \nu}
 = \frac{1}{\sqrt{|\nabla_{x'} t|^2 + 1}}
   \Big( - \langle \nabla_{x'} t, \nabla_{x'} u_0 \rangle
         + \frac{\partial u}{\partial x_n}
   \Big)
$$
on $\cS$.
\end{lemma}

\begin{proof}
This is an easy exercise.
\end{proof}

Consider a nonlinear second order partial differential equation
   $\iD u = f (x,u,\nabla u)$
in $\cX$, where
   $f (x,u,p)$
is a real analytic function on $\overline{\cX} \times \R \times \R^n$.
By Lemma \ref{l.Neumann}, the Cauchy problem for solutions of this equation
with data on $\cS$ can be formulated in the following way.
Given functions $u_0$ and
                $u_1$
on $\cS$, find a function $u$ in $\cX$ smooth up to $\cS$ which satisfies
\begin{equation}
\label{eq.Cp}
\left\{
   \begin{array}{lclcl}
     \iD u
   & =
   & f (x,u,\nabla u)
   & \mbox{in}
   & \cX,
\\
     \phantom{\iD}u
   & =
   & u_0
   & \mbox{on}
   & \cS,
\\
     \phantom{\iD}u'_{x_n}
   & =
   & u_1
   & \mbox{on}
   & \cS.
   \end{array}
\right.
\end{equation}

\begin{lemma}
\label{l.uniqueness}
There is at most one real analytic function $u$ in $\cX \cup \cS$ which is a
solution of (\ref{eq.Cp}).
\end{lemma}

\begin{proof}
Let $u_1$ and
    $u_2$
be two real analytic functions in $\cX \cup \cS$ satisfying (\ref{eq.Cp}).
Set
   $u = u_1 - u_2$,
then $u$ is real analytic in $\cX \cup \cS$ and vanishes up to the order $2$
on $\cS$.
Hence it follows that
   $\iD u = f (x,u_1,\nabla u_1) - f (x,u_2,\nabla u_2)$
vanishes on $\cS$.
Since $\iD$ is a second order elliptic operator, we readily deduce that
   $u''_{x_n x_n} = 0$
on $\cS$, and so $u$ vanishes up to order $3$ on $\cS$.
Hence it follows that $\iD u$ vanishes up to order $2$ on $\cS$, and so
   $(\partial / \partial x_n)^3 u = 0$
on $\cS$.
Arguing in this way, we conclude that $u$ vanishes up to the infinite order on
$\cS$.
Since $u$ is real analytic in $\cX \cup \cS$, we get $u \equiv 0$ in $\cX$,
   as desired.
\end{proof}

\section{Hyperbolic reduction}
\label{s.hr}

Assume that $u$ is a real analytic function in $\cX \cup \cS$ which satisfies
(\ref{eq.Cp}).
Then, for each fixed $x' \in B$, the function $u (x',x_n)$ can be extended to
a holomorphic function $u (x',x_n + \imath y_n)$ in some complex neighbourhood
of the interval $(b (x'), t (x')]$.
Without loss of generality we can assume that this neighbourhood is a triangle
$T (x')$ in the complex plane $z_n = x_n + \imath y_n$ with vertexes at
   $b (x')$ and
   $t (x') \mp \imath \varepsilon$,
where $\varepsilon > 0$ depends on $x'$.
We write $U (x',x_n,y_n)$ for the extended function, so that $u (x)$ just
amounts to $U (x',x_n,0)$.

Since $u (x',z_n)$ is holomorphic in a complex neighbourhood of
   $(b (x'), t (x')]$,
it follows from the Cauchy-Riemann equations that
$$
   \Big( \frac{\partial}{\partial x_n} \Big)^j U (x',x_n,y_n)
 = \Big( - \imath \frac{\partial}{\partial y_n} \Big)^j U (x',x_n,y_n)
$$
for all $j = 1, 2, \ldots$.
Therefore, the Cauchy problem (\ref{eq.Cp}) for $u$ transforms to the problem
\begin{equation}
\label{eq.auxiliary}
\left\{
   \begin{array}{lclcl}
     \displaystyle
     \iD_{x'} U - U''_{y_n y_n}
   & =
   & f (x',z_n,U, \nabla_{x'} U, - \imath U'_{y_n}),
   & \mbox{if}
   & x' \in B,\,
     z_n \in T (x'),
\\
     \phantom{\iD_{x'}}U (x',x_n,0)
   & =
   & u_0 (x',z_n),
   & \mbox{if}
   & x' \in B,\,
     z_n = t (x'),
\\
     \phantom{\iD_{x'}}U'_{y_n} (x',x_n,0)
   & =
   & \imath\,
     u_1 (x',z_n),
   & \mbox{if}
   & x' \in B,\,
     z_n = t (x'),
   \end{array}
\right.
\end{equation}
relative to the new unknown function $U (x',x_n,y_n)$.

Hardly can (\ref{eq.auxiliary}) be specified within Cauchy problems for second
order differential equations, for the number of independent variables is $n+1$
while the Cauchy data are given on a surface of dimension $n-1$.
Since the differential equation in (\ref{eq.auxiliary}) does not contain the
derivative $U'_{x_n}$, it is easy to deduce that the smooth solution to this
problem is by no means unique.
This no longer holds true for the holomorphic solution because of uniqueness
theorems for holomorphic functions.
Moreover, if $U (x',x_n,y_n)$ is holomorphic in $z_n = x_n + \imath y_n$, then
the differential equation in (\ref{eq.auxiliary}) is satisfied for all
   $x' \in B$ and
   $z_n \in T (x')$
provided it is fulfilled for all $x' \in B$ and
                                 $z_n = t (x') + \imath y_n$
with $|y_n| < \varepsilon$.

Thus, when one looks for a holomorphic solution to (\ref{eq.auxiliary}), this
problem actually reduces to the Cauchy problem for a quasilinear hyperbolic
equation in the space of variables $(x',y_n)$, whose principal part is given
by the wave operator.
More precisely,
\begin{equation}
\label{eq.wave}
\left\{ \! \!
   \begin{array}{rclcl}
     \displaystyle
     U''_{y_n y_n}
   & =
   & \iD_{x'} U - f (x',x_n + \imath y_n,U, \nabla_{x'} U, - \imath U'_{y_n}),
   & \mbox{if}
   & x' \in B,
\\
   &
   &
   &
   & |y_n| < \varepsilon (x'),
\\
     U (x',x_n,0)
   & =
   & u_0 (x',x_n),
   & \mbox{if}
   & x' \in B,
\\
     U'_{y_n} (x',x_n,0)
   & =
   & \imath\, u_1 (x',x_n),
   & \mbox{if}
   & x' \in B,
   \end{array}
        \! \!
\right.
\end{equation}
where the variable $x_n$ is thought of as a parameter which runs over the
interval $(b (x'),t (x'))$.
We are actually interested in the solution of this problem corresponding to
the special choice $x_n = t (x')$ of the parameter.
In other words, we study problem (\ref{eq.wave}) on the hypersurface
   $x_n = t (x')$
in the space of variables $(x,y_n)$,
   the Cauchy data being given on the intersection of the hypersurface with
   the hyperplane $\{ y_n = 0 \}$.

When passing to the Cauchy problem on the hypersurface
   $x_n = t (x')$
in $\R^{n+1}$, one should interpret equations (\ref{eq.wave}) adequately in
accordance with the presence of parameter $x_n$.
Namely, each equations has to be fulfilled together with all derivatives in
$x_n$ on $x_n = t (x')$.

\begin{lemma}
\label{l.wave}
There is at most one function
   $U (x',x_n,y_n)$
in a neighbourhood of $\cS$, which is real analytic in $y_n$ at $y_n = 0$ and
satisfies (\ref{eq.wave}) with $x_n = t (x')$.
\end{lemma}

\begin{proof}
Let $U_1$ and
    $U_2$
be two functions in a neighbourhood of $\cS$, which are real analytic in $y_n$
at $y_n = 0$ and satisfy (\ref{eq.wave}) with $x_n = t (x')$.
In the coordinates $(x',x_n,y_n)$ the surface $\cS$ is given as intersection
of two hypersurfaces
   $x_n = t (x')$, where $x' \in B$,
and
   $y_n = 0$.
Set
   $U = U_1 - U_2$,
then $U$ is real analytic in $y_n$ at $y_n = 0$.
We shall have established the lemma if we prove that each derivative
   $(\partial / \partial y_n)^j U$
with $j = 0, 1, \ldots$ vanishes for $x_n = t (x')$ and $y_n = 0$.
For $j = 0, 1$ this follows immediately from the conditions which $U_1$ and
                                                                  $U_2$
fulfil on $\cS$.
For $j \leq 2$ this follows from the differential equation in (\ref{eq.wave})
by induction.
We check it only for the initial value $j = 2$, for the induction step is
verified in much the same way.
From (\ref{eq.wave}) we get
\begin{eqnarray*}
\lefteqn{
   U''_{y_n y_n}
 \, = \,
   \iD_{x'} U_1 - \iD_{x'} U_2
}
\\
 & - &
   \left( f (x',x_n + \imath y_n, U_1, \nabla_{x'} U_1, - \imath U'_{1,y_n})
        - f (x',x_n + \imath y_n, U_2, \nabla_{x'} U_2, - \imath U'_{2,y_n})
   \right)
\end{eqnarray*}
provided that $x_n = t (x')$.

Since
$
   (\partial / \partial x_n)^j (U_1 - U_2) = 0
$
for $x_n = t (x')$,
    $y_n = 0$,
and all $j = 0, 1, \ldots$,
   it follows that
\begin{eqnarray*}
   U'_{1,x_k} (x',t (x'),0)
 & = &
   \left( U_1 (x',t (x'),0) \right)'_{x_k}
 - U_{1,x_n} (x',t (x'),0)\, t'_{x_k} (x')
\\
 & = &
   \left( U_2 (x',t (x'),0) \right)'_{x_k}
 - U_{2,x_n} (x',t (x'),0)\, t'_{x_k} (x')
\\
 & = &
   U'_{2,x_k} (x',t (x'),0)
\end{eqnarray*}
for each $k = 1, \ldots, n-1$.
Moreover, we get
\begin{equation}
\label{eq.derivatives}
   \partial_{x'}^{\alpha'} U_1
 = \partial_{x'}^{\alpha'} U_2
\end{equation}
on the surface $x_n = t (x')$,
               $y_n = 0$
for all multi-indices $\alpha' = (\alpha_1, \ldots, \alpha_{n-1})$.
This yields readily
$
   \iD_{x'} U_1
 = \iD_{x'} U_2
$
for $x_n = t (x')$ and
    $y_n = 0$.
Substituting these equalities into the formula for $U''_{y_n y_n}$ we obtain
   $U''_{y_n y_n} (x',t (x'),0) = 0$
for all $x' \in B$, as desired.
\end{proof}

Note that equalities (\ref{eq.derivatives}) generalise to
$
   \partial_{x}^{\alpha} \partial_{y_n}^{\alpha_{n+1}} U_1
 = \partial_{x}^{\alpha} \partial_{y_n}^{\alpha_{n+1}} U_2
$
for $x_n = t (x')$,
    $y_n = 0$,
and all multi-indices $\alpha = (\alpha_1, \ldots, \alpha_{n})$ and
                      $\alpha_{n+1} = 0, 1, \ldots$,
as is easy to check.

We have thus reduced the Cauchy problem for the Laplace equation perturbed by
nonlinear terms of order $\leq 1$ to the Cauchy problem for the wave equation
perturbed in the same way.
The reduction is justified as long as the solution under study is real analytic
in $x_n$.

Perhaps the reduction does not make sense in the case $n = 1$, for it leads to
no simplification.

\section{The planar case}
\label{s.tpc}

To test the hyperbolic reduction of Section \ref{s.hr}, we consider the case
$n = 2$ in detail, assuming $f$ to depend on $x \in \cX \cup \cS$ only.

Let $\cX$ be a strip domain in $\R^2$ consisting of all $x = (x_1,x_2)$, such
that $x_1 \in (a,b)$ and $b (x_1) < x_2 < t (x_1)$, where
   $(a,b)$ is a bounded interval in $\R$
and
   $b$, $t$ are smooth functions of $x_1 \in (a,b)$.
Write
   $B := (a,b)$
and denote by $\cS$ the curve $\{ (x_1,t (x_1)) : x_1 \in (a,b) \}$ which is a
part of $\partial \cX$.
We focus on the Cauchy problem for the inhomogeneous Laplace equation given by
(\ref{eq.Cp}).
When looking for a solution $u$ of this problem which extends to a holomorphic
function $u (x_1,z_2)$ of $z_2 = x_2 + \imath y_2$ in a neighbourhood of
   $\{ (x_2,0) : x_2 \in (b (x_1),t (x_1)] \}$,
for each fixed $x_1 \in (a,b)$, we arrive at
\begin{equation}
\label{eq.planar}
\left\{ \! \!
   \begin{array}{rclcl}
     \displaystyle
     U''_{y_2 y_2}
   & =
   & U''_{x_1 x_1} - f (x_1,x_2 + \imath y_2),
   & \mbox{if}
   & x_1 \in (a,b),
\\
   &
   &
   &
   & |y_2| < \varepsilon (x_1),
\\
     U (x_1,x_2,0)
   & =
   & u_0 (x_1,x_2),
   & \mbox{if}
   & x_1 \in (a,b),
\\
     U'_{y_2} (x_1,x_2,0)
   & =
   & \imath\, u_1 (x_1,x_2),
   & \mbox{if}
   & x_1 \in (a,b),
   \end{array}
        \! \!
\right.
\end{equation}
which is a Cauchy problem for the inhomogeneous wave equation with parameter
$x_2$ relative to the unknown function
   $U (x_1,x_2,y_2) = u (x_1,x_2 + \imath y_2)$,
cf. (\ref{eq.wave}).
We are actually interested in finding a function $U$ which satisfies
(\ref{eq.planar}) only on the surface $x_2 = t (x_1)$,
   see Fig.~\ref{f.n=2}.
\begin{center}
\begin{figure}[h]
\begin{picture}(100,100)
\begin{tikzpicture}
\draw[->] (0.5,0) --(3.25,0) node[right] {$x_1$};
\draw[->] (0.5,0) -- (0.5,2) node[above] {$x_2$};
\draw[->] (0.75,0.25) -- (0,-0.5) node[above] {$y_2$};

\draw[color=red] (2.5,2.5) -- (1.5,1.5);

\draw[dashed] (1.5,1.5) -- (1.5,-0.5);
\draw[dashed] (2.5,2.5) -- (2.5,0.5);

\draw (1.5,-0.5) -- (1,0);
\draw (1.5,-0.5) -- (3,0);
\draw (2.5,0.5) -- (1,0);
\draw (2.5,0.5) -- (3,0);
\draw (2.5,0.5) -- (1.5,-0.5);

\draw (2,0.82) -- (2.5,2.5);
\draw (2,0.82) -- (1.5,1.5);
\draw[dashed] (2,0) -- (2,2);
\filldraw [black] (2,0) circle (1pt);
\filldraw [black] (2,0.82) circle (1pt);
\filldraw [red] (2,2) circle (1pt) node[above=4pt] {$\mathcal{S}$}(0pt,0pt);

\filldraw [black] (0.5,0) circle (1pt);

\filldraw [black] (1,0) circle (1pt) node[below=2pt] {$a$}(0pt,0pt);
\filldraw [black] (3,0) circle (1pt) node[below=1pt] {$b$}(0pt,0pt);

\draw[dashed] (1,0) -- (1,2);
\draw[dashed] (3,0) -- (3,2);

\draw[color=red, line width=0.7pt] (1,2) .. controls
(1.5,3) and (2.5,1) .. (3,2) node[right=1pt, fill=white] {$x_2=t(x_1)$}(0pt,0pt);

\filldraw [red] (1,2) circle (1pt);
\filldraw [red] (3,2) circle (1pt);

\draw[color=black, line width=0.7pt] (3,0.85) .. controls
(2.3,0.25) and (1.3,1.5) .. (1,1);

\filldraw [black] (1,1) circle (1pt);
\filldraw [black] (3,0.85) circle (1pt) node[right=1pt] {$x_2=b(x_1)$}(0pt,0pt);

\filldraw [black] (2,1) circle (0pt) node[right=6pt, fill=white] {$\mathcal{X}$}(0pt,0pt);

\end{tikzpicture}
\end{picture}
\caption{The case $n=2$.}
\label{f.n=2}
\end{figure}
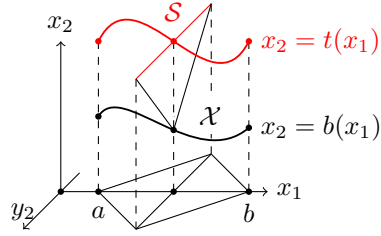
\end{center}

It is an easy exercise to verify that the function
$$
   (Gf) (x_1,x_2,y_2)
 =
 - \frac{1}{2}
   \int_{0}^{y_2}
   d y_2'
   \int_{x_1 - y_2'}^{x_1 + y_2'}
   f (x_1',x_2 + \imath (y_2-y_2'))\, d x_1'
$$
satisfies the inhomogeneous wave equation and homogeneous
   (i.e., corresponding to $u_0 = u_1 = 0$)
initial conditions in (\ref{eq.planar}).
On the hand, d'Alembert's formula gives a function satisfying the homogeneous
(i.e., corresponding to $f = 0$) wave equation and the inhomogeneous initial
conditions in (\ref{eq.planar}), see
   \cite[Ch.~I, \S~7.1]{CourHilb68}.
In fact, this is
\begin{equation}
\label{eq.d'Alembert}
   P (u_0,u_1) (x_1,x_2,y_2)
 =
   \frac{u_0 (x_1 \! + \! y_2,x_2) + u_0 (x_1 \! - \! y_2,x_2)}{2}
 + \frac{\imath}{2}
   \int_{x_1 - y_2}^{x_1 + y_2}
   u_1 (x_1',x_2) d x_1',
\end{equation}
where the right-hand side is well defined for all $(x_1,x_2,y_2)$ satisfying
   $x_1 + y_2 \in (a,b)$ and
   $x_1 - y_2 \in (a,b)$.
The pairs $(x_1,y_2)$ with this property form two cones $C^\pm$ in the plane,
$C^\pm$ being the set of all $(x_1,y_2)$, such that
   $x_1 \in (a,b)$ and
   $\pm y_2 \in [0,\varepsilon (x_1))$,
where
$$
   \varepsilon (x_1)
 = \frac{b-a}{2} - \Big| x_1 - \frac{a+b}{2} \Big|.
$$

Thus, given any twice differentiable function $u_0 (x_1,x_2)$,
                differentiable function $u_1 (x_1,x_2)$ of $x_1 \in (a,b)$
and
            any differentiable function $f (x_1,z_2)$ of both variables,
the formula
$$
   U = Gf + P (u_0,u_1)
$$
yields a solution to the Cauchy problem (\ref{eq.planar}) for all values of
parameter $x_2$ that do not lead beyond the domains of $u_0$, $u_1$ and $f$.
Had we known $u_0 (x_1,x_2)$ and
             $u_1 (x_1,x_2)$
for all values $x_2 \in (b (x_1),t (x_1)]$, then the first initial condition
of (\ref{eq.planar}) would give
   $U (x_1,x_2,0) = u_0 (x_1,x_2)$
and so the solution to the Cauchy problem (\ref{eq.Cp}) by
   $u (x) = u_0 (x_1,x_2)$.
This just recovers the reduction but is not of use to solve the original
Cauchy problem.
However, on substituting $x_2 = t (x_1)$ into $U (x_1,x_2,y_2)$ we obtain
\begin{eqnarray}
\label{eq.roaf}
\lefteqn{
   u (x_1, t (x_1) + \imath y_2)
 \, = \,
 - \frac{1}{2}
   \int_{0}^{y_2}
   d y_2'
   \int_{x_1 - y_2'}^{x_1 + y_2'}
   f (x_1',t (x_1) + \imath (y_2-y_2'))\, d x_1'
}
\nonumber
\\
 & + &
   \frac{u_0 (x_1 \! + \! y_2,t (x_1)) + u_0 (x_1 \! - \! y_2,t (x_1))}{2}
 + \frac{\imath}{2}
   \int_{x_1 - y_2}^{x_1 + y_2}
   u_1 (x_1',t (x_1)) d x_1'
\nonumber
\\
\end{eqnarray}
for all $x_1 \in (a,b)$ and
        $|y_2| < \varepsilon (x_1)$.
Note that
   $(x_1', t (x_1))$
fails to lie on the curve $\cS$ for all $x_1' \in [x_1 - y_2,x_1 + y_2]$
   unless $t (x_1)$ is constant.
Therefore, $u (x_1, t (x_1) + \imath y_2)$ is determined by the Cauchy data of
$u$ in some neighbourhood of $\cS$.
This forces us once again to confine ourselves with solutions which are real
analytic in the variable $x_2$.

For fixed $x_1 \in (a,b)$, formula (\ref{eq.roaf}) gives the restriction of
the function $u (x_1,z_2)$, holomorphic in $z_2$ in the triangle with vertexes
at
   $b (x_1)$ and
   $t (x_1) \mp \imath \varepsilon (x_1)$,
to the side
   $t (x_1) + \imath [-\varepsilon (x_1), \varepsilon (x_1)]$
of the triangle.
This limits application of hyperbolic theory.
Our next objective is to continue the function from the side of the triangle
analytically along the bisectrix of the angle at $b (x_1)$.
This is a problem of analytic continuation.

\section{Carleman formula}
\label{s.cf}

Let $D$ be a domain in the complex plane $\C$ of variable $z$ bounded by lines
$BO$ and
$OA$ and
by a smooth curve $c = AB$ lying inside the angle $BOA$.
Write
   $\angle BOA = \alpha \pi$
with $0 < \alpha < 2$.

Choose the univalent branch of the analytic function $\sqrt[\alpha]{w}$ in the
complex plane with a slit along the ray $\arg w = \pi$, which takes the value
$1$ at $w = 1$.

\begin{lemma}
\label{l.Carleman}
If $u$ is a holomorphic function in $D$ continuous up to the boundary, then
$$
   u (z)
 = \lim_{N \to \infty}
   \frac{1}{2 \pi \imath}
   \int_c
   u (\zeta)\,
   \exp N \Big( \Big( \frac{\zeta - \zeta_0}{z - \zeta_0} \Big)^{1/\alpha} - 1
          \Big)
   \frac{d \zeta}{\zeta - z}
$$
holds for any point $z \in D$ on the bisectrix of the angle $BOA$, where
   $\zeta_0$ is a complex number corresponding to the vertex $O$ of the angle.
\end{lemma}

This formula is due to Carleman \cite{Carl26}.
To our best knowledge it was the first formula of analytic continuation using
the idea of quenching function.
Since that time such formulas in complex analysis and elliptic theory are
called Carleman formulas, see \cite{Aize93},
                              \cite{Tark95}.

\begin{proof}
Fix any $z \in D$ lying on the bisectrix of the angle $BOA$.
For $N = 1, 2, \ldots$, we apply the Cauchy integral formula to the function
$$
   u (\zeta)\,
   \exp N \Big( \Big( \frac{\zeta - \zeta_0}{z - \zeta_0} \Big)^{1/\alpha} - 1
          \Big)
$$
which is holomorphic in $D$ and continuous in the closure of $D$.
Since its value at $\zeta = z$ is $u (z)$, we get
\begin{eqnarray}
\label{eq.Cauchy}
   u (z)
 & = &
   \frac{1}{2 \pi \imath}
   \int_c
   u (\zeta)\,
   \exp N \Big( \Big( \frac{\zeta - \zeta_0}{z - \zeta_0} \Big)^{1/\alpha} - 1
          \Big)
   \frac{d \zeta}{\zeta - z}
\nonumber
\\
 & + &
   \frac{1}{2 \pi \imath}
   \int_{\partial D \setminus c}
   u (\zeta)\,
   \exp N \Big( \Big( \frac{\zeta - \zeta_0}{z - \zeta_0} \Big)^{1/\alpha} - 1
          \Big)
   \frac{d \zeta}{\zeta - z}.
\end{eqnarray}
If $\zeta \in \partial D \setminus c$, then
\begin{eqnarray*}
   \Big( \frac{\zeta - \zeta_0}{z - \zeta_0} \Big)^{1/\alpha}
 & = &
   \Big| \frac{\zeta - \zeta_0}{z - \zeta_0} \Big|^{1/\alpha}
   \exp \Big( \pm \frac{\pi}{2} \imath \Big)
\\
 & = &
   \pm \Big| \frac{\zeta - \zeta_0}{z - \zeta_0} \Big|^{1/\alpha} \imath
\end{eqnarray*}
and so the modulus of
$
   \displaystyle
   \exp N \Big( \Big( \frac{\zeta - \zeta_0}{z - \zeta_0} \Big)^{1/\alpha} - 1
          \Big)
$
equals $e^{-N}$.
Letting $N \to \infty$ in (\ref{eq.Cauchy}) establishes the lemma.
\end{proof}

Having disposed of this preliminary step, we now turn to the problem of
analytic continuation we have encountered in Section \ref{s.tpc}.
We apply Lemma \ref{l.Carleman} in the plane of complex variable
   $z_2 = x_2 + \imath y_2$.
Given any fixed $x_1 \in (a,b)$, we take the triangle $T (x_1)$ with vertexes
   $O := b (x_1)$ and
   $A := t (x_1) - \imath \varepsilon (x_1)$,
   $B := t (x_1) + \imath \varepsilon (x_1)$
as $D$,
   cf. Fig.~\ref{f.Carleman}.
\begin{center}
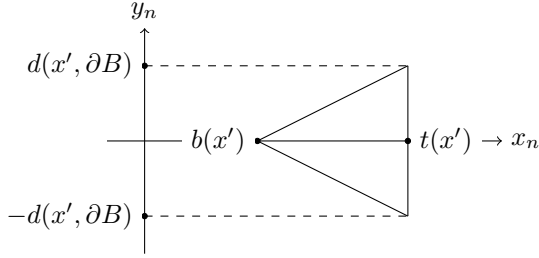
\begin{figure}[h]
\begin{picture}(100,100)
\begin{tikzpicture}

\draw[->] (-1,0) --(4.25,0) node[right] {$x_n$};
\draw[->] (-0.5,-1.5) -- (-0.5,1.5) node[above] {$y_n$};
\draw (3,-1) -- (3,1);
\draw (1,0) -- (3,1);
\draw (1,0) -- (3,-1);

\draw[dashed] (-0.5,1) -- (3,1);
\draw[dashed] (-0.5,-1) -- (3,-1);
\filldraw [black] (-0.5,1) circle (1pt) node[left=1pt, fill=white] {$d(x',\partial B)$};
\filldraw [black] (-0.5,-1) circle (1pt) node[left=1pt, fill=white] {$-d(x',\partial B)$};

\filldraw [black] (1,0) circle (1pt) node[left=1pt, fill=white] {$b(x')$};
\filldraw [black] (3,0) circle (1pt) node[right=1pt, fill=white] {$t(x')$};

\end{tikzpicture}
\end{picture}
\caption{Recovering a holomorphic function.}
\label{f.Carleman}
\end{figure}
\end{center}

In this case
$$
   \alpha = \frac{2}{\pi}\, \arctan \Big( \frac{\varepsilon (x_1)}
                                               {t (x_1) - b (x_1)}
                                    \Big)
$$
depends on $x_1$ and the bisectrix of the angle $BOA$ coincides with the real
axis.
The solution $u (x_1,z_2)$ is given on the edge $AB$ and we are aimed at
reconstructing it in the interval $(b (x_1),t (x_1))$.

\begin{theorem}
\label{t.d'Alembert}
Let $n = 2$.
For each solution $u$ of the Cauchy problem (\ref{eq.Cp}) in $\cX$ which is
real analytic up to $\cS$, the formula
$$
   u (x)
 = \lim_{N \to \infty}
   \frac{1}{2 \pi}
   \! \! \!
   \int\limits_{- \varepsilon (x_1\!)}^{\varepsilon (x_1\!)}
   \! \! \!
   U (x_1, t (x_1\!), y_2)
   \exp N \Big( \! \Big( \frac{t (x_1\!) \! - \! b (x_1\!) \! + \! \imath y_2}
                              {x_2 - b (x_1\!)}
                \Big)^{\frac{\scriptstyle 1}{\scriptstyle \alpha}} \! - \! 1 \!
          \Big)
   \frac{d y_2}{t (x_1\!) \! - \! x_2 \! + \! \imath y_2}
$$
holds for all $x \in \cX$.
\end{theorem}

\begin{proof}
This follows immediately from Lemma \ref{l.Carleman} and
                              formula (\ref{eq.roaf})
giving an explicit continuation of the solution $u (x_1,x_2)$ along $\cS$ to
the plane of complex variable
   $z_2 = x_2 + \imath y_2$.
\end{proof}

This formula is especially simple if $\cS$ is a segment $x_2 = t_0$, i.e. the
graph of a constant function $t (x_1) \equiv t_0$ of $x_1 \in (a,b)$.
If moreover $f \equiv 0$ then formula (\ref{eq.roaf}) transforms to
$$
   U (x_1, t_0, y_2)
 =
   \frac{u_0 (x_1 \! + \! y_2,t_0) + u_0 (x_1 \! - \! y_2,t_0)}{2}
 + \frac{\imath}{2}
   \int_{x_1 - y_2}^{x_1 + y_2}
   u_1 (x_1',t_0) d x_1'
$$
for all $x_1 \in (a,b)$ and
        $|y_2| < \varepsilon (x_1)$.
Substituting this into the formula of Theorem \ref{t.d'Alembert} we get
\begin{eqnarray}
\label{eq.Carleman+d'Alembert}
   u (x)
 & = &
   \lim_{N \to \infty}
   \!\!
   \int\limits_{x_1 - \varepsilon (x_1)}^{x_1 + \varepsilon (x_1)}
   \!\!
   u (x_1',t_0)\,
   \Re\, K_N (x_1,x_2,x_1-x_1')
   \,
   d x_1'
\nonumber
\\
 & - &
   \lim_{N \to \infty}
   \!\!
   \int\limits_{x_1 - \varepsilon (x_1)}^{x_1 + \varepsilon (x_1)}
   \!\!
   \frac{\partial u}{\partial x_2} (x_1',t_0)
   \Big(
   \int\limits_{|x_1'-x_1|}^{\varepsilon (x_1)}
   \Im\, K_N (x_1,x_2,y_2)
   \,
   d y_2
   \Big)
   d x_1',
\nonumber
\\
\end{eqnarray}
where
$$
   K_N (x',x_n,y_n)
 = \frac{1}{2 \pi}
   \frac{\displaystyle
         \exp N \Big( \Big( \frac{t (x') - b (x') + \imath y_n}
                                 {x_n - b (x')}
                       \Big)^{\frac{\scriptstyle 1}{\scriptstyle \alpha}} - 1
                \Big)
        }
        {t (x') - x_n + \imath y_n}.
$$

Formula (\ref{eq.Carleman+d'Alembert}) can be regarded as an elliptic analogue
of the d'Alembert formula for the wave equation.

Note that nowadays there are many explicit formulas of analytic continuation
which are simpler than the original formula of \cite{Carl26}.
We refer the reader to \cite{Aize93}.

\section{Poisson formula}
\label{s.pf}

In this section we discuss the case $n = 3$ in detail, assuming the function
$f$ to depend on $x \in \cX \cup \cS$ only.
The Cauchy problems for the inhomogeneous Laplace equation reduces to the
Cauchy problem for the inhomogeneous wave equation.
This latter reads
\begin{equation}
\label{eq.n=3}
\left\{ \! \!
   \begin{array}{rclcl}
     \displaystyle
     U''_{y_3 y_3}
   & =
   & \iD_{x'} - f (x',x_3 + \imath y_3),
   & \mbox{if}
   & x' \in B,
\\
   &
   &
   &
   & |y_3| < \varepsilon (x'),
\\
     U (x',x_3,0)
   & =
   & u_0 (x',x_3),
   & \mbox{if}
   & x' \in B,
\\
     U'_{y_3} (x',x_3,0)
   & =
   & \imath\, u_1 (x',x_3),
   & \mbox{if}
   & x' \in B,
   \end{array}
        \! \!
\right.
\end{equation}
$x_3$ being thought of as parameter.
We are aimed at finding a function $U$ which fulfills (\ref{eq.n=3}) on the
surface $x_3 = t (x')$.

The advantage of the reduction lies in the fact that the Cauchy problem for
hyperbolic equations is well posed in the class of smooth functions.
For $n = 3$, there is an explicit formula for its solution due to Poisson,
   see \cite[Ch.~III, \S~6.5]{CourHilb68}.
More precisely,
\begin{eqnarray}
\label{eq.Poisson}
\lefteqn{
   U (x', x_3, y_3)
 \, = \,
 - \frac{1}{2 \pi}
   \int\limits_{0}^{y_3}
   d y_3'
   \int\limits_{|x'' - x'| < |y_3'|}
   \frac{f (x'', x_3 + \imath (y_3-y_3'))}
        {\sqrt{y_3'{}^2 - |x''-x'|^2}}\,
   d x''
}
\nonumber
\\
 \! & \! + \! & \!
   \frac{\partial}{\partial y_3}\,
   \frac{\mathrm{sgn}\, y_3}{2 \pi}
   \! \! \! \!
   \int\limits_{|x'' - x'| < |y_3|}
   \! \! \! \!
   \frac{u_0 (x'', x_3)}
        {\sqrt{y_3^2\!-\!|x''-x'|^2}}
   d x''
 + \frac{\mathrm{sgn}\, y_3}{2 \pi}
   \! \! \! \!
   \int\limits_{|x'' - x'| < |y_3|}
   \! \! \! \!
   \frac{\imath u_1 (x'', x_3)}
        {\sqrt{y_3^2\!-\!|x''-x'|^2}}
   d x''
\nonumber
\\
\end{eqnarray}
for all $x' \in B$ and
        $|y_3| < \varepsilon (x')$.

For formula (\ref{eq.Poisson}) to make sense it is certainly required that,
   for any $y_3$,
the ball $|x''-x'| < |y_3|$ would belong to the domain $B$ in $\R^{n-1}_{x'}$,
where the Cauchy data $u_0 (x',x_n)$ and
                      $u_1 (x',x_n)$
are given.
Since $y_3$ varies in the interval $(-\varepsilon (x'), \varepsilon (x'))$, we
get readily the formula
   $\varepsilon (x') = d (x',\partial B)$,
the distance from $x'$ to the boundary of $B$,
   cf. Fig.~\ref{f.cones}.
\begin{center}
\begin{figure}[h]
\begin{picture}(120,120)

\begin{tikzpicture}

\draw[->] (0,0) --(3,0) node[right] {$ $};
\draw[->] (0,0) -- (0,1.5) node[above] {$y_n$};
\draw[->] (0,0) -- (-0.7,-0.7) node[right] {$ $};

\draw (3,-0.75) arc (0:360:1.5cm and 0.5cm) node[right=1pt, fill=white] {$B$}(0pt,0pt);
\draw[dashed] (1.5,-0.75) arc (0:180:0.75cm and 0.25cm);
\draw (0,-0.75) arc (180:360:0.75cm and 0.25cm);

\filldraw [black] (0.75,-0.75) circle (1pt);

\draw[dashed] (0.75,0) -- (0.75,-1.5) node[below=2pt, fill=white] {$x'$};
\draw (0.75,-1.5) -- (0, -0.75);
\draw (0.75,-1.5) -- (1.5, -0.75);
\draw (0.75,0) -- (0, -0.75);
\draw (0.75,0) -- (1.5, -0.75);
\draw[dashed] (0,0) -- (0.75,-0.75);
\draw[dashed] (0,0.75) -- (0.75,0);

\filldraw [black] (0,0.75) circle (1pt) node[left=1pt, fill=white] {$d(x',\partial B)$};

\end{tikzpicture}
\end{picture}
\caption{Reduction to imaginary cones.}
\label{f.cones}
\end{figure}
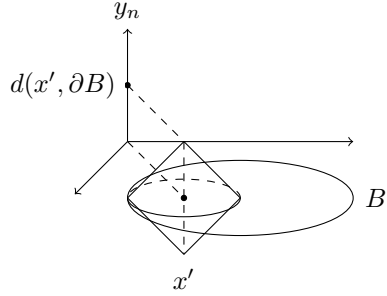
\end{center}

\begin{theorem}
\label{t.Poisson}
Let $n = 3$.
For each solution $u$ of the Cauchy problem (\ref{eq.Cp}) in $\cX$ which is
real analytic up to $\cS$, the formula
$$
   u (x)
 = \lim_{N \to \infty}
   \frac{1}{2 \pi}
   \! \! \!
   \int\limits_{- \varepsilon (x')}^{\varepsilon (x')}
   \! \! \!
   U (x', t (x'), y_3)
   \exp N \Big( \! \Big( \frac{t (x') \! - \! b (x') \! + \! \imath y_3}
                              {x_3 - b (x')}
                   \Big)^{\frac{\scriptstyle 1}{\scriptstyle \alpha}}\!-\!1\!
          \Big)
   \frac{d y_3}{t (x') \! - \! x_3 \! + \! \imath y_3}
$$
holds for all $x \in \cX$, where
$
   \displaystyle
   \alpha = \frac{2}{\pi}\, \arctan \Big( \frac{\varepsilon (x')}
                                               {t (x') - b (x')}
                                    \Big).
$
\end{theorem}

\begin{proof}
This is a direct consequence of Lemma \ref{l.Carleman} and
                                formula (\ref{eq.Poisson})
which gives an explicit continuation of the solution $u (x',x_3)$ along $\cS$
to the plane of complex variable
   $z_3 = x_3 + \imath y_3$.
\end{proof}

On substituting (\ref{eq.Poisson}) into
                the Carleman formula of Theorem \ref{t.Poisson}
we arrive at an explicit formula for solutions of the Cauchy problem for the
inhomogeneous Laplace equation.
The computations are cumbersome, and so we confine ourselves with the case
$f \equiv 0$, as in (\ref{eq.Carleman+d'Alembert}).
By the very construction of the Carleman kernel,
   $K_N (x', x_3, \varepsilon (x'))$ tends to zero as $N \to \infty$,
for any $x' \in B$ and
        $x_3 \in (b (x'), t (x'))$.
Hence
\begin{eqnarray}
\label{eq.Carleman+Poisson}
   u (x)
 & = &
   -\,
   \lim_{N \to \infty}
   \!\!
   \int\limits_{|x''-x'| < \varepsilon (x')}
   \!\!
   u (x'',t (x'))
   \Big(
   \int\limits_{|x''-x'|}^{\varepsilon (x')}
   \frac{1}{\pi} \frac{\displaystyle
                       \frac{\partial}{\partial y_3} \Re\, K_N (x',x_3,y_3)}
                      {\sqrt{y_3^2 - |x''-x'|^2}}
   \,
   d y_3
   \Big)
   d x''
\nonumber
\\
 & &
   -\,
   \lim_{N \to \infty}
   \!\!
   \int\limits_{|x''-x'| < \varepsilon (x')}
   \!\!
   \frac{\partial u}{\partial x_3} (x'',t (x'))
   \Big(
   \int\limits_{|x''-x'|}^{\varepsilon (x')}
   \frac{1}{\pi} \frac{\Im\, K_N (x',x_3,y_3)}
                      {\sqrt{y_3^2 - |x''-x'|^2}}
   \,
   d y_3
   \Big)
   d x''
\nonumber
\\
\end{eqnarray}
for all $x \in \cX$.

Formula (\ref{eq.Carleman+Poisson}) can be thought of as an elliptic analogue
of the Poisson formula for the wave equation.

\section{Kirchhoff formula}
\label{s.kf}

The solution of the Cauchy problem for the wave equation bears certain
structure which changes in odd and even dimensions.
For this reason we consider also the case $n = 4$ in detail.
The corresponding formula for solutions of the Cauchy problem for the wave
equations is known as the Kirchhoff formula,
   see \cite[Ch.~III, \S~6.4]{CourHilb68} and elsewhere.

By the above, the Cauchy problem for the Laplace equation in a cylindrical
domain $\cX \subset \R^4$ reduced to
\begin{equation}
\label{eq.n=4}
\left\{ \! \!
   \begin{array}{rclcl}
     \displaystyle
     U''_{y_4 y_4}
   & =
   & \iD_{x'} - f (x',x_4 + \imath y_4),
   & \mbox{if}
   & x' \in B,
\\
   &
   &
   &
   & |y_4| < \varepsilon (x'),
\\
     U (x',x_4,0)
   & =
   & u_0 (x',x_4),
   & \mbox{if}
   & x' \in B,
\\
     U'_{y_4} (x',x_4,0)
   & =
   & \imath\, u_1 (x',x_4),
   & \mbox{if}
   & x' \in B,
   \end{array}
        \! \!
\right.
\end{equation}
where
   $x' = (x_1,x_2,x_3)$ varies in a domain $B \subset \R^3$,
   $\varepsilon (x')$ stands for the distance from $x' \in B$ to the boundary
   of $B$,
and
   $x_3$ is thought of as parameter in $(b (x'), t (x')]$.
The Cauchy data $u_0$ and $u_1$ are in $C^3 (B)$ and $C^2 (B)$, respectively.
The Kirchhoff formula gives
\begin{eqnarray}
\label{eq.Kirchhoff}
\lefteqn{
   U (x', x_4, y_4)
 \, = \,
 - \frac{1}{4 \pi}
   \int\limits_{|x'' - x'| < |y_4|}
   \frac{f (x'', x_4 + \imath (y_4-|x''-x'|))}
        {|x''-x'|}\,
   d x''
}
\nonumber
\\
 \! & \! + \! & \!
   \frac{\partial}{\partial y_4}\,
   \frac{1}{4 \pi y_4}
   \! \! \!
   \int\limits_{|x'' - x'| = |y_4|}
   \! \! \!
   u_0 (x'', x_4)
   d \sigma (x'')
 + \frac{1}{4 \pi y_4}
   \! \! \!
   \int\limits_{|x'' - x'| = |y_4|}
   \! \! \!
   \imath u_1 (x'', x_4)
   d \sigma (x'')
\nonumber
\\
\end{eqnarray}
for all $x' \in B$ and
        $|y_4| < \varepsilon (x')$.

The substitution $x_4 = t (x')$ into $U$ gives the restriction of the function
$U$, holomorphic in $z_4 = x_4 + \imath y_4$, to the edge
   $t (x') + \imath [- \varepsilon (x'), \varepsilon (x')]$
of the triangle $T (x') \subset \C$, where $U$ is holomorphic.
Using Carleman's formula of Lemma \ref{l.Carleman}, we arrive at a formula for
$u (x)$ similar to that of Theorem \ref{t.Poisson}.
It reads in much the same way, with $x_3$ and $y_3$ replaced by $x_4$ and
                                                                $y_4$,
respectively.
For short we restrict our attention to a formula like
   (\ref{eq.Carleman+Poisson}).

\begin{corollary}
\label{c.Kirchhoff}
Let $n = 4$.
For each solution $u$ of the Cauchy problem (\ref{eq.Cp}) with $f \equiv 0$ in
$\cX$, which is real analytic up to $\cS$, we get
\begin{eqnarray}
\label{eq.Carleman+Kirchhoff}
   u (x)
 & = &
   -\,
   \lim_{N \to \infty}
   \!\!
   \int\limits_{|x''-x'| < \varepsilon (x')}
   \!\!
   u (x'',t (x'))
   \,
   \frac{1}{2 \pi} \frac{\displaystyle
                         \Big( \frac{\partial}{\partial y_4} \Re\, K_N \Big)
                         (x',x_4,|x''-x'|)}
                         {|x''-x'|}
   \,
   d x''
\nonumber
\\
 & &
   -\,
   \lim_{N \to \infty}
   \!\!
   \int\limits_{|x''-x'| < \varepsilon (x')}
   \!\!
   \frac{\partial u}{\partial x_4} (x'',t (x'))\,
   \frac{1}{2 \pi} \frac{\Im\, K_N (x',x_4,|x''-x'|)}
                        {|x''-x'|}
   \,
   d x''
\nonumber
\\
\end{eqnarray}
for all $x \in \cX$.
\end{corollary}

\begin{proof}
The proof is quite elementary although cumbersome.
We first substitute the integral of $u_0$ on the left-hand side of
(\ref{eq.Kirchhoff}) into Carleman's formula.
Integration by parts yields
\begin{eqnarray*}
\lefteqn{
   \int\limits_{- \varepsilon (x')}^{\varepsilon (x')}
   \frac{\partial}{\partial y_4}
   \Big(
   \frac{1}{4 \pi y_4}
   \int\limits_{|x'' - x'| = |y_4|}
   u_0 (x'', t (x'))
   d \sigma (x'')
   \Big)
   K_N (x', x_4, y_4)\,
   d y_4
}
\\
 & = &
   \Big(
   \frac{1}{4 \pi y_4}
   \int\limits_{|x'' - x'| = |y_4|}
   u_0 (x'', t (x'))
   d \sigma (x'')
   \Big)
   K_N (x', x_4, y_4)\,
   \Big|_{y_4 = - \varepsilon (x')}^{y_4 = + \varepsilon (x')}
\\
 & - &
   \int\limits_{- \varepsilon (x')}^{\varepsilon (x')}
   \Big(
   \frac{1}{4 \pi y_4}
   \int\limits_{|x'' - x'| = |y_4|}
   u_0 (x'', t (x'))
   d \sigma (x'')
   \Big)
   \frac{\partial}{\partial y_4}
   K_N (x', x_4, y_4)\,
   d y_4.
\end{eqnarray*}

The first integral on the right-hand side is equal to
$$
   \Big(
   \frac{1}{2 \pi \varepsilon (x')}
   \int\limits_{|x'' - x'| = \varepsilon (x')}
   u_0 (x'', t (x'))
   d \sigma (x'')
   \Big)
   \Re\, K_N (x', x_4, \varepsilon (x')),
$$
which vanishes as $N \to \infty$ by the construction of the kernel
   $K_N (x', x_4, \varepsilon (x'))$.
Indeed, the point $t (x') + \imath \varepsilon (x')$ belongs to the top leg
of the angle $BOA$, and $x_4$ to its bisectrix.

Furthermore, we write the second integral on the right-hand side as the sum of
two integrals.
The first integral is over $y_4 \in (-\varepsilon (x'),0)$ and the second one
                      over $y_4 \in (0,\varepsilon (x'))$.
In the second integral we change the variable by $y_4 \mapsto - y_4$, and then
evaluate the sum, obtaining
\begin{eqnarray*}
\lefteqn{
   -\,
   \int\limits_{- \varepsilon (x')}^{\varepsilon (x')}
   \Big(
   \frac{1}{4 \pi y_4}
   \int\limits_{|x'' - x'| = |y_4|}
   u_0 (x'', t (x'))
   d \sigma (x'')
   \Big)
   \frac{\partial}{\partial y_4}
   K_N (x', x_4, y_4)\,
   d y_4
}
\\
 & = &
   -\,
   \int\limits_{0}^{\varepsilon (x')}
   \Big(
   \frac{1}{2 \pi y_4}
   \int\limits_{|x'' - x'| = |y_4|}
   u_0 (x'', t (x'))
   d \sigma (x'')
   \Big)
   \frac{\partial}{\partial y_4}
   \Re\, K_N (x', x_4, y_4)\,
   d y_4.
\end{eqnarray*}
Since $d x'' = d \sigma (x'') d y_4$, we deduce from Fubini's theorem that the
latter integral just amounts to
$$
   -\,
   \int\limits_{|x''-x'| < \varepsilon (x')}
   \!\!
   u (x'',t (x'))
   \,
   \frac{1}{2 \pi} \frac{\displaystyle
                         \Big( \frac{\partial}{\partial y_4} \Re\, K_N \Big)
                         (x',x_4,|x''-x'|)}
                         {|x''-x'|}
   \,
   d x'',
$$
as desired.

The same (even easier) reasoning applies when one substitutes the integral of
$u_1$ on the left-hand side of (\ref{eq.Kirchhoff}) into Carleman's formula.
The details are left to the reader.
\end{proof}

Formula (\ref{eq.Carleman+Kirchhoff}) is an exposition of Kirchhoff's formula
for the wave equation in the context of elliptic theory.
We have already mentioned another interpretation of Kirchhoff's formula in
   \cite{Kryl69}.
Unfortunately, we could not understand this latter paper.

\section{Concluding remarks}
\label{s.cr}

The developed method of analytic continuation in the plane of complex variable
   $z_n = x_n + \imath y_n$
still works if the Cauchy problem under study is nonlinear.
Having granted a holomorphic solution $U (x',x_n,y_n)$ to the Cauchy problem
(\ref{eq.wave}) on the surface $x_n = t (x')$, we use Carleman's formula to
extend $U$ to all of $\cX$.
The extension looks like
\begin{equation}
\label{eq.nlp}
   u (x)
 = \lim_{N \to \infty}
   \int_{- \varepsilon (x')}^{\varepsilon (x')}
   U (x', t (x'), y_n)\,
   K_N (x',x_n,y_n)\,
   d y_n
\end{equation}
for all $x  \in \cX$.

Formula (\ref{eq.nlp}) allows one to construct explicit formulas similar to
   (\ref{eq.Carleman+d'Alembert}),
   (\ref{eq.Carleman+Poisson}) and
   (\ref{eq.Carleman+Kirchhoff})
for arbitrary $n$.
To this end one uses classical formulas for the solution of the Cauchy problem
for a second order hyperbolic equation by the descent method of Hadamard,
   cf. \cite{Hada23}, \cite[Ch.~VI, \S.~5.2]{CourHilb68}.
We were rather interested in equations of mathematical physics.

The simplest formula is obtained for even $n \geq 4$,
   thus generalising Kirchhoff's formula (\ref{eq.Carleman+Kirchhoff}).
If
   $u_0 \in C^{(n+2)/2} (\cS)$ and
   $u_1 \in C^{n/2} (\cS)$,
then every solution $u$ of (\ref{eq.Cp}) with $f \equiv 0$ represents by
\begin{eqnarray}
\label{eq.Carleman+neven}
\lefteqn{
   u (x)
 \, = \,
   \lim_{N \to \infty}
   \int\limits_{|x''-x'| < \varepsilon (x')}
   d x''
}
\nonumber
\\
 & &
   u (x'',t (x'))
   \,
   \frac{(-1)^{\scriptstyle \frac{n}{2} - 1}\, 2}
        {\sigma_{n\!-\!1} 1 \! \cdot \! 3 \cdot \! \ldots \! \cdot (n\!-\!3)}
   \frac{\displaystyle
         \Big( \Big( \frac{\partial}{\partial y_n} \frac{1}{y_n}
               \Big)^{\scriptstyle \frac{n-2}{2}}
         y_n\, \Re\, K_N
         \Big)
         (x',x_n,|x''\!-\!x'|)}
        {|x''-x'|}
\nonumber
\\
 & + &
   \frac{\partial u}{\partial x_4} (x'',t (x'))
   \,
   \frac{(-1)^{\scriptstyle \frac{n}{2} - 1}\, 2}
        {\sigma_{n\!-\!1} 1 \! \cdot \! 3 \cdot \! \ldots \! \cdot (n\!-\!3)}
   \frac{\displaystyle
         \Big( \Big( \frac{\partial}{\partial y_n} \frac{1}{y_n}
               \Big)^{\scriptstyle \frac{n-4}{2}}
         \Im\, K_N
         \Big)
         (x',x_n,|x''\!-\!x'|)}
        {|x''-x'|}
\nonumber
\\
\end{eqnarray}
for all $x \in \cX$, where
   $\sigma_{n-1}$ stands for the area of the $(n-2)\,$-dimensional unit sphere
   in $\R^{n-1}$.
We used here an exotic designation for the integral by purely technical
reasons.

\begin{remark}
\label{r.Yarmukhamedov}
Formula (\ref{eq.Carleman+neven}) has much in common with the familiar formula
of \cite{Yar75}.
\end{remark}

The method of proof carries over to right-hand sides $f (x,u,\nabla u)$ which
are affine functions of $u$ and $\nabla u$.
This is the case, e.g., for the Helmholtz equation,
   cf. \cite[Ch.~VI, \S.~5.7]{CourHilb68}.

Another class of equations which may be handled in much the same way consists
of those of the form
$$
   Au + u''_{x_n x_n} = f (x),
$$
where $A$ is a linear differential operator containing at most the derivative
$u'_{x_n}$ but no higher order derivatives in $x_n$,
   see \cite[Ch.~III, \S~6.4]{CourHilb68}.

\bigskip

\textit{Acknowledgments\,}
The research of the first author was done in the framework of the
   Mikhail Lomonosov Fellowship
which is supported by the Russian Ministry of Education and
                      the Deutsche Forschungsgemeinschaft.

\end{document}